\documentclass[nocopyrightspace, preprint]{llncs}

\usepackage{amsmath,amssymb,latexsym}

\usepackage{amsmath}

\usepackage{comment}
\usepackage{amsfonts}
\usepackage{amssymb}
\usepackage{mathtools}
\usepackage{xspace}
\usepackage{commands}
\usepackage{hyperref}
\usepackage{listings}

\title{Verified Parallel String Matching in Haskell}
\author{Niki Vazou\inst{1,2} \and Jeff Polakow\inst{2}}
\institute{$^1$UC San Diego \quad $^2$Awake Networks}

\newcommand\stringMempty{\ensuremath{\eta}}
\newcommand\stringMappend{\ensuremath{\boxdot}}
\newcommand\mempty{\ensuremath{\epsilon}}
\newcommand\mappend{\ensuremath{\diamondsuit}}

\newcommand\tx{\ensuremath{\texttt{x}}\xspace}
\newcommand\ty{\ensuremath{\texttt{y}}\xspace}
\newcommand\txs{\ensuremath{\texttt{xs}}\xspace}

\newtheorem{assumption}[theorem]{Assumption}

\usepackage{listings}
\usepackage{amssymb}

\ifdefined\withcolor
	 \definecolor{haskellblue}{rgb}{0.0, 0.0, 1.0}
	 \definecolor{haskellstr}{rgb}{0.2, 0.2, 0.6}
	 \definecolor{haskellred}{rgb}{1.0, 0.0, 0.0}
  \definecolor{gray_ulisses}{gray}{0.55}
  \definecolor{castanho_ulisses}{rgb}{0.71,0.33,0.14}
  \definecolor{preto_ulisses}{rgb}{0.41,0.20,0.04}
  \definecolor{green_ulises}{rgb}{0.2,0.75,0}
\else
	\definecolor{haskellblue}{gray}{0.1}
	\definecolor{haskellstr}{gray}{0.1}
	\definecolor{haskellred}{gray}{0.1}
	\definecolor{gray_ulisses}{gray}{0.1}
	\definecolor{castanho_ulisses}{gray}{0.1}
	\definecolor{preto_ulisses}{gray}{0.1}
	\definecolor{green_ulisses}{gray}{0.1}
\fi

\definecolor{lcolor}{gray}{0.0}
\definecolor{lappcolor}{gray}{0.0}
\definecolor{lappascolor}{gray}{0.0}

\def\codesize{\normalsize}

\newcommand\showfocus[1]{\color{purple}{\textbf{#1}}}

\lstdefinelanguage{HaskellUlisses} {
	basicstyle=\ttfamily\codesize,
	moredelim=[is][\showfocus]{\#}{\#},
	sensitive=true,
	morecomment=[l][\color{gray_ulisses}\ttfamily\itshape\codesize]{--},
	morecomment=[s][\color{gray_ulisses}\ttfamily\itshape\codesize]{\{-}{-\}},
	morestring=[b]",
	stringstyle=\color{haskellstr},
	showstringspaces=false,
	numberstyle=\codesize,
	numberblanklines=true,
	showspaces=false,
	breaklines=true,
	showtabs=false,
  literate={
           {`}{{{$^{\backprime}{}$}}}1
           {'}{{{$^{\prime}{}$}}}1
           {?}{{{$\therefore$}}}1
           {<=}{{$\leq$}}1
           {theta}{{$\theta$}}1
           {rf}{{{\color{lappcolor}f}}}1
           {gf}{{{\color{lappascolor}f}}}1
           {rmap}{{{\color{lappcolor}map}}}3
           {gmap}{{{\color{lappascolor}map}}}3
           {env}{{$\Gamma$}}1
           {|-}{{$\vdash$}}1
           {<=!}{{{\color{lcolor}<=!}}}3
           {!=}{{$\neq$}}1
           {forall}{{$\forall$}}1
           {->}{{$\rightarrow$}}1
           {=*}{{$\eqfun$}}2
           {<=>}{{$\Leftrightarrow$}}3
           {=>}{{$\Rightarrow$}}2
           {<:}{{$\preceq$}}1
           {mempty}{{$\mempty$}}1
           {mappend}{{$\mappend$}}1
           {<>}{{$\mappend$}}1
           {stringMempty}{{$\stringMempty$}}1
           {<+>}{{$\stringMappend$}}1
           {stringMappend}{{$\stringMappend$}}1
           {listMempty}{{[]}}1
           {listMappend}{{++}}2
           {epsilon}{{$\epsilon$}}1
           {eta}{{$\eta$}}1
           {&&&}{&&&}3
           {&&}{{$\land$}}1
           {_m}{{${}_m$}}1
           {_n}{{${}_n$}}1
           {m^+}{{m${}^{+}$}}2
           },
	emph=
	{[1]
		FilePath,IOError,abs,acos,acosh,all,and,any,appendFile,approxRational,asTypeOf,asin,
		asinh,atan,atan2,atanh,basicIORun,break,catch,ceiling,chr,compare,concat,concatMap,
		const,cos,cosh,curry,cycle,decodeFloat,denominator,digitToInt,div,divMod,drop,
		dropWhile,either,elem,encodeFloat,enumFrom,enumFromThen,enumFromThenTo,enumFromTo,
		error,even,exp,exponent,fail,mapMaybe,filter,flip,floatDigits,floatRadix,floatRange,floor,
		fmap,foldl,foldl1,foldr,foldr1,fromDouble,fromEnum,fromInt,fromInteger,fromIntegral,
		fromRational,fst,gcd,getChar,getContents,getLine,head,id,inRange,index,init,intToDigit,
		interact,ioError,isAlpha,isAlphaNum,isAscii,isControl,isDenormalized,isDigit,isHexDigit,
		isIEEE,isInfinite,isLower,isNaN,isNegativeZero,isOctDigit,isPrint,isSpace,isUpper,iterate,
		last,lcm,length,lex,lexDigits,lexLitChar,lines,log,logBase,lookup,map,mapM,mapM_,max,
		maxBound,posMax,negMax,maximum,maybe,min,minBound,minimum,mod,negate,not,notElem,null,numerator,odd,
		or,ord,pi,pred,primExitWith,print,product,properFraction,putChar,putStr,putStrLn,quot,
		quotRem,range,rangeSize,read,readDec,readFile,readFloat,readHex,readIO,readInt,readList,readLitChar,
		readLn,readOct,readParen,readSigned,reads,readsPrec,realToFrac,recip,rem,repeat,replicate,return,
		reverse,round,scaleFloat,scanl,scanl1,scanr,scanr1,seq,sequence,sequence_,show,showChar,showInt,
		showList,showLitChar,showParen,showSigned,showString,shows,showsPrec,significand,signum,sin,
		sinh,snd,span,splitAt,sqrt,subtract,succ,sum,tail,take,takeWhile,tan,tanh,threadToIOResult,toEnum,
		toInt,toInteger,toLower,toRational,toUpper,truncate,uncurry,undefined,unlines,until,unwords,unzip,
		unzip3,userError,words,writeFile,zip,zip3,zipWith,zipWith3,listArray,doParse,empty,for,initTo,
        assert,compose,checkGE,maxEvens,empty,create,get,set,initialize,idVec,fastFib,fibMemo,
        ex1,ex2,ex3,incr,inc,dec,isPos,positives,find,insert,len,size,union,fromList,initUpto,trim,
        insertSort,decsort,qsort,reverse,append,upperCase, ifM, whileM, get, decrM, diff,
        project, select, leq, elts, keys, dkeys, dfun, addKey, pTrue, emptyRD, rFalse,
        	dom, rng, isI, isD, isS, movie1, movie2,  toI, toS, toD, good_titles, runState, ret,
        	update, getCtr, setCtr, ctr, rdCtr, wrCtr, ifTest, whileTest, posCtr, zeroCtr, decr, decCtr,
        	pread , pwrite , plookup , pcontents, pcreateF , pcreateFP, pcreateD, active, caps, pset, eqP,
        	write, contents, alloc, derivP, copyP, createDir, store, copyRec, copySpec,
        	forM_, when, flookup, fread, createDir, pcreateFile, isFile, copyFrame, ?
	},
	emphstyle={[1]\color{haskellblue}},
	emph=
	{[2] 	Show,Eq,Ord,Num,UpClosed,Comp,Wit,Witness,Inductive,Meet,Flip,TRUE,Nat,Pos,Neg,IntGE,Plus,List,
        Bool,Char,Double,Either,Float,IO,Integer,Int,Maybe,
        Ordering,Rational,Ratio,ReadS,ShowS,String,Word8,
        InPacket,Tree,Vec,NullTerm,IncrList,DecrList,
        UniqList,BST,MinHeap,MaxHeap,World,RIO,IO,HIO,Post,Pre,
        Privilege, Prop, Chain, ChainTy, Range, Dict, RD, Dom, Set, P, Univ, Schema, MovieSchema, RT,
        TDom, TRange, MoviesTable, RTSubEqFlds, RTEqFlds, Disjoint, Union, Ret, Seq, Trans, Map,
        Pure, Then, Else, Exit, Inv, OneState, Priv, Path, FH, Stable,
				Prop, Nat,
	},
	emphstyle={[2]\color{castanho_ulisses}},
	emph=
	{[3]
		case,class,data,deriving,do,else,if,import,in,infixl,infixr,instance,let,
		module,of,primitive,then,refinement,type,where,forall,bound,
		measure,reflect,predicate
	},
	emphstyle={[3]\color{preto_ulisses}\textbf},
	emph=
	{[4]
		quot,rem,div,mod,elem,notElem,seq
	},
	emphstyle={[4]\color{castanho_ulisses}\textbf},
	emph=
	{[5]
		EQ,GT,LT,Left,Right
	},
	emphstyle={[5]\color{preto_ulisses}\textbf},
	emph=
	{[6]
	    axiomatize, measure, inline
	},
	emphstyle={[6]\color{lcolor}}
}

\lstnewenvironment{code}
{\lstset{language=HaskellUlisses}}
{}

\lstnewenvironment{mcode}
{\lstset{language=HaskellUlisses,columns=fullflexible,keepspaces,mathescape}}
{}

\lstMakeShortInline[language=HaskellUlisses,mathescape,keepspaces,mathescape,basicstyle=\ttfamily\normalsize,breakatwhitespace]@

\lstdefinelanguage{Pseudo} {
	basicstyle=\ttfamily\codesize,
	sensitive=true,
  mathescape=true,
	morecomment=[l][\color{gray_ulisses}\ttfamily\codesize]{--},
	morecomment=[s][\color{gray_ulisses}\ttfamily\codesize]{\{-}{-\}},
	morestring=[b]",
	showstringspaces=false,
	numberstyle=\codesize,
	numberblanklines=true,
	showspaces=false,
	breaklines=true,
	showtabs=false
}

\begin{document}
\maketitle


\begin{abstract}
In this paper, we prove correctness of parallelizing a string matcher
using Haskell as a theorem prover.
We use refinement types to specify correctness properties,
Haskell terms to express proofs and Liquid Haskell to
check correctness of proofs.
First, we specify and prove that a class of monoid morphisms
can be parallelized via parallel monoid concatenation.
Then, we encode string matching as a morphism
to get a provably correct parallel transformation.
Our 1839LoC prototype proof shows that
Liquid Haskell can be used as a fully expressive
theorem prover on realistic Haskell implementations.

\end{abstract}

\section{Introduction}\label{sec:intro}
In this paper, we prove correctness of parallelization of a na\"ive string matcher
using Haskell as a theorem prover.
We use refinement types to specify correctness properties,
Haskell terms to express proofs and Liquid Haskell to
check correctness of proofs.

Optimization of sequential functions via parallelization
is a well studied technique~\cite{jaja,blelloch}.
Paper and pencil proofs of program have been developed to support the
correctness of the transformation~\cite{Cole93parallelprogramming}.
However, these paper and pencil proofs
show correctness of the parallelization algorithm
and do not reason about the actual implementation
that may end up being buggy.

Dependent Type Systems (like Coq~\cite{coq-book} and Adga~\cite{agda} )
enable program equivalence proofs
for the actual implementation of the functions to be parallelized.
For example, SyDPaCC~\cite{SyDPaCC} is a Coq extension that
given a na\"ive Coq implementation of a function,
returns an Ocaml parallelized version with a proof of program equivalence.
The limitation of this approach is that the initial
function should be implemented in
the specific dependent type framework
and thus cannot use features and libraries from one's favorite
programming language.

Refinement Types~\cite{ConstableS87,FreemanPfenningDONTCITE91,Rushby98}
on the other hand, enable verification of existing general purpose languages
(including
ML~\cite{pfenningxi98,GordonRefinement09,LiquidPLDI08},
C~\cite{deputy,LiquidPOPL10},
Haskell~\cite{Vazou14},
Racket~\cite{RefinedRacket}
and Scala~\cite{refinedscala}).
Traditionally, refinement types are limited to
``shallow'' specifications,
that is, they are used to specify and verify properties
that only talk about behaviors of program functions
but not functions themselves.
This restriction critically limits the expressiveness
of the specifications
but allows for automatic SMT~\cite{SMTLIB2} based verification.
Yet, program equivalence proofs were out of the expressive
power of refinement types.

Recently, we extended refinement types
with Refinement Reflection~\cite{reflection},
a technique that reflects each function's implementation
into the function's type
and is implemented in Liquid Haskell~\cite{Vazou14}.
We claimed that Refinement Reflection can turn any programming
language into a proof assistant.
In this paper we check our claim and use Liquid Haskell
to prove program equivalence.
Specifically,
we \textit{define in Haskell}
a sequential string matching function, @toSM@, and its parallelization, @toSMPar@,
using existing Haskell libraries; then,
we \textit{prove in Haskell} that these two functions are equivalent,
and we check our proofs using Liquid Haskell.

\paragraph{Theorems as Refinement Types}
Refinement Types refine types with properties
drawn from decidable logics.
For example, the type @{v:Int | 0 < v}@
describes all integer values @v@ that are greater than @0@.
We refine the unit type to express theorems,
define unit value terms to express proofs, and use
Liquid Haskell to check that the proofs prove the theorems.
For example, Liquid Haskell accepts
the type assignment @() :: {v:()| 1+1=2}@,
as the underlying SMT can always prove the equality @1+1=2@.
We write @{1+1=2}@ to simplify the type @{v:()| 1+1=2}@
from the irrelevant binder @v:()@.

\paragraph{Program Properties as Types}
The theorems we express can refer to program functions.
As an example, the type of @assoc@ expresses that @mappend@
is associative.
\begin{code}
assoc :: x:m -> y:m -> z:m -> {x mappend (y mappend z) = (x mappend y) mappend z}
\end{code}
In \S~\ref{sec:haskell-proofs} we explain
how to write Haskell proof terms to prove theorems like @assoc@
by proving that list append @(++)@ is associative.
Moreover, we prove that the empty list @[]@ is the identity element of
list append, and conclude that the list
(with @[]@ and @(++)@, \ie the triple (@[a]@, @[]@, @(++)@))
is provably a monoid.

\paragraph{Corectness of Parallelization}
In \S~\ref{sec:parallelization}, we define the type @Morphism n m f@ that specifies
that @f@ is a \textit{morphism} between two monoids
(@n@, @$\eta$@, @<+>@) and (@m@, @$\epsilon$@, @<>@),
\ie @f :: n -> m@ where @f $\eta$ = $\epsilon$@ and @f (x <+> y) = f x <> f y@.

A morphism @f@ on a ``chunkable'' input type can be parallelized by:
\begin{enumerate}
  \item chunking up the input in @j@ chunks (@chunk j@),
  \item applying the morphism in parallel to all chunks (@pmap f@), and
  \item recombining the mapped chunks via @mappend@, also in parallel (@pmconcat i@).
\end{enumerate}
We specify correctness of the above transformation as a refinement type.
\begin{code}
parallelismEquivalence
  :: f:(n -> m) -> Morphism n m f -> x:n -> i:Pos -> j:Pos
   -> {f x = pmconcat i (pmap f (chunk j x))}
\end{code}
\S~\ref{sec:parallelization} describes the parallelization transformation in details,
while Correctness of Parallelization Theorem~\ref{theorem:two-level} proves correctness
by providing a
term that satisfies the above type.

\paragraph{Case Study: Parallelization of String Matching}
We use the above theorem to parallelize string matching.
We define a string matching function @toSM :: RString -> toSM target@
from a \textit{refined string} to a \textit{string matcher}.
A refined string (\S~\ref{subsec:refinedstrings}) is a wrapper around
the efficient string manipulation library
@ByteString@ that moreover assumes
various string properties, including the monoid laws.
A string matcher @SM target@ (\S~\ref{subsec:stringmatcher}) is a data type that contains
a refined string and a list of all the indices
where the type level symbol @target@ appears in the input.
We prove that @SM target@ is a monoid and @toSM@ is a morphism,
thus by the aforementioned Correctness of Parallelization Theorem~\ref{theorem:two-level}
we can correctly parallelize string matching.

To sum up, we present the first realistic proof that 
uses Haskell as a theorem prover:
correctness of parallelization on string matching.
Our contributions are summarized as follows
\begin{itemize}
\item We explain how theorems and proofs are encoded and checked in Liquid Haskell
by formalizing monoids and proving that lists are monoids
(\S~\ref{sec:haskell-proofs}).
\item We formalize morphisms between monoids and
specify and prove correctness of parallelization of morphisms
(\S~\ref{sec:parallelization}).
\item We show how libraries can be imported as trusted components by wrapping
@ByteString@s as refined strings which satisfy the monoid laws (\S~\ref{subsec:refinedstrings}).
\item As an application, we prove that a string matcher is a morphism between the monoids of refined strings
and string matchers,  thus we get provably correct parallelization of string matching (\S~\ref{sec:stringmatching}).
\item Based on our 1839LoC proof we evaluate the approach of using Haskell as a theorem prover
(\S~\ref{sec:evaluation}).
\end{itemize}

\section{Proofs as Haskell Functions}\label{sec:haskell-proofs}

Refinement Reflection~\cite{reflection} is a technique
that lets you write Haskell functions that prove theorems
about other Haskell functions and have your proofs machine-checked
by Liquid Haskell~\cite{Vazou14}.
As an introduction to Refinement Reflection,
in this section, we prove that lists are monoids by
\begin{itemize}
\item \textit{specifying monoid laws} as refinement types,
\item \textit{proving the laws} by writing the implementation of the law specifications, and
\item \textit{verifying the proofs} using Liquid Haskell.
\end{itemize}

\subsection{Reflection of data types into logic.}
To start with,
we define a List data structure and
teach Liquid Haskell basic properties about List,
namely, how to check that proofs on lists are \textit{total}
and how to encode functions on List into the logic.

The data list definition @L@ is the standard recursive definition.
\begin{code}
data L [length] a = N | C a (L a)
\end{code}
With the @length@ annotation in the definition Liquid Haskell
will use the @length@ function to check termination
of functions recursive on Lists.
We define @length@ as the standard Haskell function
that returns natural numbers.
We lift @length@ into logic as a \textit{measure}~\cite{Vazou14},
that is, a \textit{unary} function whose (1) domain is the data type, and
(2) body is a single case-expression over the datatype.
\begin{code}
type Nat = {v:Int | 0 <= v}

measure length
length         :: L a -> Nat
length N        = 0
length (C x xs) = 1 + length xs
\end{code}

Finally, we teach Liquid Haskell how to encode functions on Lists
into logic.
The flag @"--exact-data-cons"@
automatically derives measures which
(1) test if a value has a given data constructor, and
(2) extract the corresponding field's value.
For example, Liquid Haskell will automatically derive the following
List manipulation measures from the List definition.
\begin{code}
isN :: L a -> Bool    -- Haskell's null
isC :: L a -> Bool    -- Haskell's not . null

selC1 :: L a -> a     -- Haskell's head
selC2 :: L a -> L a   -- Haskell's tail
\end{code}
Next, we describe how Liquid Haskell uses the above measures
to automatically reflect Haskell functions on Lists into logic.

\subsection{Reflection of Haskell functions into logic.}
Next, we define and reflect into logic the two monoid operators on Lists.
Namely, the identity element @mempty@ (which is the empty list)
and an associative operator @(mappend)@ (which is list append).
\begin{code}
reflect mempty
mempty :: L a
mempty = N

reflect (mappend)
(mappend) :: L a -> L a -> L a
N        mappend ys = ys
(C x xs) mappend ys = C x (xs mappend ys)
\end{code}

The reflect annotations lift the Haskell functions into logic in three steps.
First, check that the Haskell functions indeed terminate by checking
that the @length@ of the input list is decreasing,
as specified in the data list definition.
Second, in the logic, they define the respective uninterpreted functions
@mempty@ and @(mappend)@.
Finally, the Haskell functions and the logical uninterpreted functions
are related by strengthening the result type of the Haskell function
with the definition of the function's implementation.
For example, with the above @reflect@ annotations,
Liquid Haskell will \textit{automatically} derive the following strengthened
types for the relevant functions.
\begin{code}
mempty  :: {v:L a | v = mempty && v = N }

(mappend):: xs:L a -> ys:L a
    ->  {v:L a | v = xs mappend ys
              && v = if isN xs then ys
                     else C (selC1 xs) (selC2 xs mappend ys)
        }
\end{code}

\subsection{Specification and Verification of Monoid Laws}

Now we are ready to specify the monoid laws as refinement types and
provide their respective proofs as terms of those type. Liquid Haskell
will verify that our proofs are valid. 
Note that this is exactly what
one would do in any standard logical framework,
like LF~\cite{Harper93}.

The type @Proof@ is defined as an alias of the unit type (@()@)
in the library @ProofCombinators@ that comes with Liquid Haskell.
Figure~\ref{figure:proofcombinators} summarizes the definitions we use from @ProofCombinators@.
We express theorems as refinement types by refining
the @Proof@ type with appropriate refinements.
For example, the following theorem states
the @mempty@ is always equal to itself.
\begin{code}
trivial :: {mempty = mempty}
\end{code}
Where @{mempty = mempty}@ is a simplification for the @Proof@ type
@{v:Proof | mempty = mempty}@, since the binder @v@ is irrelevant, and
@trivial@ is defined in @ProofCombinators@ to be unit.
Liquid Haskell will typecheck the above code using an SMT
solver to check congruence on @mempty@.
%

\begin{definition}[Monoid] \label{definition:monoid}
The triple (@m@, @epsilon@, @<>@) is a monoid
(with identity element @epsilon@ and associative operator @<>@),
if the following functions are defined. 
\begin{code}
idLeft_m :: x:m -> {mempty mappend x = x}
idRight_m :: x:m -> {x mappend mempty = x}
assoc_m :: x:m -> y:m -> z:m -> {x mappend (y mappend z) = (x mappend y) mappend z}
\end{code}
\end{definition}
Using the above definition, we prove that our list type @L@ is a monoid
by defining Haskell proof terms that satisfy the above monoid laws. 

\paragraph{Left Identity} is expressed
as a refinement type signature that takes as input
a list @x:L a@ and returns a @Proof@ type refined
with the property @mempty <> x = x@
\begin{code}
idLeft :: x:L a -> { mempty mappend x = x }
idLeft x = empty <> x ==. N <> x ==. x *** QED
\end{code}
We prove left identity using combinators from @ProofCombinators@ as
defined in Figure~\ref{figure:proofcombinators}.
We start from the left hand side @empty <> x@,
which is equal to @N <> x@ by calling @empty@ thus
unfolding the equality @empty = N@ into the logic.
Next, the call @N <> x@ unfolds into the logic the definition of @(<>)@
on @N@ and @x@, which is equal to @x@, concluding our proof.
Finally, we use the operators @p *** QED@ which basically casts @p@ into a proof term.
In short, the proof of left identity, proceeds by unfolding the definitions of @mempty@
and @(<>)@ on the empty list.

\begin{figure}[t]
\begin{code}
type Proof = ()
data QED   = QED

trivial :: Proof
trivial = ()

(==.) :: x:a -> y:{a | x = y} -> {v:a | v = x}
x ==. _ = x

(***) :: a -> QED -> Proof
_ *** _ = ()

(?) :: (Proof -> a) -> Proof -> a
f ? y = f y
\end{code}
\caption{Operators and Types defined in \texttt{ProofCombinators}}
\label{figure:proofcombinators}
\end{figure}

\paragraph{Right identity} is proved by structural induction.
We encode inductive proofs by case splitting on the base and inductive case,
and enforcing the inductive hypothesis via a recursive call.
\begin{code}
idRight :: x:L a -> { x <> mempty = x }
idRight N = N <> empty ==. N *** QED

idRight (C x xs)
   =   (C x xs) <> empty
   ==. C x (xs <> empty)
   ==. C x xs ? idRight xs
   *** QED
\end{code}
The recursive call @idRight xs@ is provided
as a third optional argument in the @(==.)@
operator to justify the equality @xs <> empty = xs@,
while the operator @(?)@ is merely a function application
with the appropriate precedence.
Note that LiquiHaskell, via termination and totality checking,
is verifying that all the proof terms are well formed because
(1) the inductive hypothesis is only applying to smaller terms, and
(2) all cases are covered.

\paragraph{Associativity} is proved in a very similar manner,
using structural induction.
\begin{code}
assoc   :: x:L a -> y:L a -> z:L a
        -> { x mappend (y mappend z) = (x mappend y) mappend z}
assoc N y z
  =   N <> (y <> z)
  ==. y <> z
  ==. (N <> y) <> z
  *** QED

assoc (C x xs) y z
  =  (C x xs) <> (y <> z)
  ==. C x (xs <> (y <> z))
  ==. C x ((xs <> y) <> z) ? associativity xs y z
  ==. (C x (xs <> y)) <> z
  ==. ((C x xs) <> y) <> z
  *** QED
 \end{code}
As with the left identity, the proof proceeds by
(1) function unfolding (or rewriting in paper and pencil proof terms),
(2) case splitting (or case analysis), and
(3) recursion (or induction).

Since our list implementation satisfies the three monoid laws
we can conclude that @L a@ is a monoid.
%


\begin{theorem}\label{theorem:monoid:list}
(@L a@, @epsilon@, @<>@) is a monoid.
\end{theorem}
\begin{proof}
@L a@ is a monoid, as the implementation of
@idLeft@, @idRight@, and @assoc@
satisfy the specifications of
@idLeft_m@, @idRight_m@, and @assoc_m@, with @m = L a@.
\qed\end{proof}

\section{Verified Parallelization of Monoid Morphisms}\label{sec:parallelization}

A monoid morphism is a function between two monoids which
preserves the monoidal structure; \ie a function on the underlying
sets which preserves identity and associativity. We formally specify
this definition using a refinement type @Morphism@.
\begin{definition}[Monoid Morphism]\label{definition:morphism}
A function @f :: n -> m@ is a morphism
between the monoids
(@m@, @$\epsilon$@, @<>@)
and (@n@, @$\eta$@, @<+>@)
if @Morphism n m f@ has an inhabitant.
\begin{code}
type Morphism n m F =
  x:n -> y:n -> {F eta = epsilon && F (x <+> y) = F x <> F y}
\end{code}
\end{definition}

A monoid morphism can be parallelized when its domain can be cut into
chunks and put back together again, a property we refer to as
chunkable and expand upon in \S~\ref{subsec:chunkable}. A
chunkable monoid morphism is then parallelized by:
\begin{enumerate}
  \item chunking up the input,
  \item applying the morphism in parallel to all chunks, and
  \item recombining the chunks, also in parallel, back to a single value.
\end{enumerate}
In the rest of this section we implement and verify to be correct the above
transformation.

\subsection{Chunkable Monoids}\label{subsec:chunkable}
\begin{definition}[Chunkable Monoids]\label{definition:chunkable}
A monoid (@m@, @epsilon@, @<>@) is chunkable
if the following four functions are defined on @m@.
\begin{code}
length_m :: m -> Nat

drop_m :: i:Nat -> x:MGEq m i -> MEq m (length_m x-i)
take_m :: i:Nat -> x:MGEq m i -> MEq m i

takeDropProp_m :: i:Nat -> x:m ->
                  {x = take_m i x <> drop_m i x}
\end{code}

Where the type aliases @MLeq m I@ (and @MEq m I@)
constrain the monoid @m@ to have @length_m@
greater than (resp. equal) to @I@.
\begin{code}
type MGEq m I = {x:m | I <= length_m x }
type MEq  m I = {x:m | I = length_m x }
\end{code}
\end{definition}

Note that the ``important'' methods of chunkable monoids
are the @take@ and @drop@, while the @length@ method is required
to give pre- and post-condition on the other operations.
Finally, @takeDropProp@ provides a proof that
for each @i@ and monoid @x@, appending
@take i x@ to @drop i x@ will reconstruct @x@.

Using @take_m@ and @drop_m@ we define for each chunkable monoid
(@m@, @epsilon@, @<>@) a function @chunk_m i x@ that
splits @x@ in chunks of size @i@.
\begin{code}
chunk_m :: i:Pos -> x:m -> {v:L m | chunkRes_m i x v }
chunk_m i x
  | length_m x <= i = C x N
  | otherwise     = take_m i x `C` chunk_m i (drop_m i x)

chunkRes_m i x v
  | length_m x <= i = length_m v == 1
  | i == 1        = length_m v == length_m xs
  | otherwise     = length_m v < length_m xs
\end{code}

The function @chunk_m@ provably terminates as
@drop_m i x@
will return a monoid smaller than @x@,
by the Definition of @drop_m@.
The definitions of both @take_m@ and @drop_m@
are also used from Liquid Haskell to verify the
@length_m@ constraints in the result of @chunk_m@.

\ignore{
As a concrete example, to define list chunking, we first define the @take@ and @drop@
methods on the list monoid of section~\ref{sec:haskell-proofs}.
\begin{code}
take i N                    = N
take i (C x xs) | i == 0    = N
                | otherwise = C x (take (i-1) xs)

drop i N                    = N
drop i (C x xs) | i == 0    = C x xs
                | otherwise = drop (i-1) xs
\end{code}
We can prove that the above definitions
combined with the @length@ of section~\ref{sec:haskell-proofs}
satisfy the specifications
of the Chunkable Monoid Definition~\ref{definition:chunkable}.
Thus, we can prove that the aforementioned list data type,
extended with the appropriate implementation for @takeDropProp@
is a chunkable monoid.
}

\subsection{Parallel Map}
We define a parallelized map function @pmap@
using Haskell's library @parallel@.
Concretely, we use the function
@Control.Parallel.Strategies.withStrategy@
that computes its argument in parallel given a parallel strategy.
\begin{code}
pmap :: (a -> b) -> L a -> L b
pmap f xs = withStrategy parStrategy (map f xs)
\end{code}
The strategy @parStrategy@ does not affect verification.
In our codebase we choose the traversable strategy.
\begin{code}
parStrategy :: Strategy (L a)
parStrategy = parTraversable rseq
\end{code}

\paragraph{Parallelism in the Logic.}
The function @withStrategy@ is an imported Haskell library function,
whose implementation is not available during verification.
To use it in our verified code, we make the \textit{assumption}
that it always returns its second argument.
\begin{code}
assume withStrategy :: Strategy a -> x:a -> {v:a | v = x}
\end{code}
Moreover, we need to reflect the function @pmap@ and represent its
implementation in the logic.
Thus, we also need to represent the function @withStrategy@ in the logic.
LiquidHaskell represents @withStrategy@ in the logic as a logical
function that merely returns
its second argument, @withStrategy _ x = x@,
and does not reason about parallelism.

\subsection{Monoidal Concatenation}\label{subsec:mconcat}
The function @chunk_m@ lets us turn a monoidal value into several
pieces. In the other direction, for any monoid @m@, there is a
standard way of turning @L m@ back into a single @m@~\footnote{\texttt{mconcat} is usually defined as \texttt{foldr mappend mempty}}
\begin{code}
  mconcat :: L m -> m
  mconcat N        = mempty
  mconcat (C x xs) = x <> mconcat xs
\end{code}
For any chunkable monoid @n@,
monoid morphism @f :: n -> m@,
and natural number @i > 0@
we can write a chunked version of @f@ as
\begin{code}
  mconcat . pmap f . chunk_n i :: n -> m.
\end{code}
Before parallelizing @mconcat@, we will prove that the previous function is equivalent to @f@.

\begin{theorem}[Morphism Distribution]\label{theorem:monoid:distribution}
Let (@m@, @$\epsilon$@, @<>@) be a monoid
and (@n@, @$\eta$@, @<+>@) be a chunkable monoid.
Then, for every morphism @f :: n -> m@,
every positive number @i@ and input @x@,
@f x = mconcat (pmap f (chunk_n i x))@ holds.
\begin{code}
morphismDistribution
  :: f:(n -> m) -> Morphism n m f -> x:n -> i:Pos
  -> {f x = mconcat (pmap f (chunk_n i x))}
\end{code}
\end{theorem}

\begin{proof}
We prove the theorem by providing an implementation of
@morphismDistribution@ that satisfies its type.
The proof proceeds by induction on the length of the input.
\begin{code}
morphismDistribution f thm x i
  | length_n x <= i
  =   mconcat (pmap f (chunk_n i x))
  ==. mconcat (map f (chunk_n i x))
  ==. mconcat (map f (C x N))
  ==. mconcat (f x `C` map f N)
  ==. f is <> mconcat N
  ==. f is <> epsilon
  ==. f is ? idRight_m (f is)
  *** QED
morphismDistribution f thm x i
  =   mconcat (pmap f (chunk_n i x))
  ==. mconcat (map f (chunk_n i x))
  ==. mconcat (map f (C takeX) (chunk_n i dropX)))
  ==. mconcat (f takeX `C` map f (chunk_n n dropX))
  ==. f takeX <> f dropX
      ? morphismDistribution f thm dropX i
  ==. f (takeX <+> dropX)
      ? thm takeX dropX
  ==. f x
      ? takeDropProp_n i x
  *** QED
  where
    dropX = drop_n i x
    takeX = take_n i x
\end{code}
In the base case we use rewriting and right identity on the monoid @f x@.
In the inductive case,
we use the inductive hypothesis on the input @dropX = drop_n i x@,
that is provably smaller than @x@ as @1 < i@.
Then, the fact that @f@ is a monoid morphism,
as encoded by our assumption argument @thm takeX dropX@
we get basic distribution of @f@, that is
@f takeX <> f dropX = f (takeX <+> dropX)@.
Finally, we merge @takeX <+> dropX@ to @x@
using the property @takeDropProp_n@ of the chunkable monoid @n@.
\qed\end{proof}

\subsection{Parallel Monoidal Concatenation}\label{subsec:pmconcat}
We now parallelize the monoid concatenation by defining a
@pmconat i x@ function that chunks the input list of monoids and concatenates each
chunk in parallel.

We use the @chunk@ function of \S~\ref{subsec:chunkable} instantiated to @L m@ to define a parallelized version of
monoid concatenation @pmconcat@.
\begin{code}
pmconcat :: Int -> L m -> m
pmconcat i x | i <= 1 || length x <= i
  = mconcat x
pmconcat i x
  = pmconcat i (pmap mconcat (chunk i x))
\end{code}
The function @pmconcat i x@ calls @mconcat x@ in the base case,
otherwise it
(1) chunks the list @x@ in lists of size @i@,
(2) runs in parallel @mconcat@ to each chunk,
(3) recursively runs itself with the resulting list.
Termination of @pmconcat@ holds, as the length of @chunk i x@
is smaller than the length of @x@, when @1 < i@.

Next, we prove equivalence of parallelized monoid concatenation.
\begin{theorem}[Correctness of Parallelization]\label{theorem:equivalence:concat}
Let (@m@, @$\epsilon$@, @<>@) be a monoid.
Then, the parallel and sequential concatenations are equivalent.
\begin{code}
pmconcatEquivalence
  :: i:Int -> x:L m -> { pmconcat i x = mconcat x }
\end{code}
\end{theorem}

\begin{proof}
We prove the theorem by providing a Haskell implementation of @pmconcatEquivalence@
that satisfies its type.
The details of the proof can be found in~\cite{implementation},
here we provide the sketch of the proof.

First, we prove that @mconcat@ distributes over list splitting
\begin{code}
mconcatSplit
  :: i:Nat -> xs:{L m | i <= length xs}
  -> { mconcat xs = mconcat (take i xs)
                 <> mconcat (drop i xs) }
\end{code}
The proofs proceeds by structural induction, using monoid left identity in the base case
and monoid associativity associavity and unfolding of @take@ and @drop@
methods in the inductive step.

We generalize the above lemma
to prove that @mconcat@ distributes over list chunking.
\begin{code}
mconcatChunk
  :: i:Pos -> xs:L m
  -> { mconcat xs = mconcat (map mconcat (chunk i xs)) }
\end{code}
The proofs proceeds by structural induction, using monoid left identity in the base case
and lemma @mconcatSplit@ in the inductive step.

Lemma @mconcatChunk@ is sufficient to prove @pmconcatEquivalence@ by structural induction,
using monoid left identity in the base case.
\qed\end{proof}

\subsection{Parallel Monoid Morphism}\label{subsec:both-levels}
We can now replace the @mconcat@ in our chunked monoid morphism in
\S~\ref{subsec:mconcat} with @pmconcat@ from
\S~\ref{subsec:pmconcat} to provide an implementation that uses
parallelism to both map the monoid morphism and concatenate the
results.

%
\begin{theorem}[Correctness of Parallelization]\label{theorem:two-level}
Let (@m@, @$\epsilon$@, @<>@) be a monoid
and (@n@, @$\eta$@, @<+>@) be a chunkable monoid.
Then, for every morphism @f :: n -> m@,
every positive numbers @i@ and @j@, and input @x@,
@f x = pmconcat i (pmap f (chunk_n j x))@ holds.
\begin{code}
parallelismEquivalence
  :: f:(n -> m) -> Morphism n m f -> x:n -> i:Pos -> j:Pos
  -> {f x = pmconcat i (pmap f (chunk_n j x))}
\end{code}
\end{theorem}

\begin{proof}
We prove the theorem by providing an implementation of
@parallelismEquivalence@ that satisfies its type.
\begin{code}
parallelismEquivalence f thm x i j
  =   pmconcat i (pmap f (chunk_n j x))
  ==. mconcat (pmap f (chunk_n j x))
      ? pmconcatEquivalence i (pmap f (chunk_n j x))
  ==. f x
      ? morphismDistribution f thm x j
  *** QED
\end{code}
The proof follows merely by application of the 
two previous Theorems~\ref{theorem:monoid:distribution} and~\ref{theorem:equivalence:concat}.
\qed\end{proof}

\ignore{
\paragraph{A Basic Time Complexity} analysis of the algorithm
reveals that parallelization of morphism  leads to runtime speedups
on monads with fast (constant time) appending operator.

We want to compare the complexities of the sequential @f i@
and the two-level parallel @pmconcat i (pmap f (chunk_n j x))@.
Let $n$ be the size on the input @x@.
Then, the sequential version runs in time
$T_f(n) = O(n)$, that is equal to the time complexity of the morphism @f@ on input @i@.

The parallel version runs @f@ on inputs of size $n' = \frac{n}{j}$.
Assuming the complexity of @x <> y@ to be $T_\mappend(\text{max}(|\tx|, |\ty|))$,
complexity of @mconcat xs@ is $O((\texttt{length \txs}-1) T_\mappend(\text{max}_{\tx_i \in \txs}(|\tx_i|)))$.
Now, parallel concatenation, @pmconcat i xs@ at each iteration runs @mappend@
on a list of size @i@. Moreover,
at each iteration, divides the input list in chunks of size @i@, leading to
$\frac{\log|xs|}{\log i}$ iterations, and time complexity
$(i-1)(\frac{\log|xs|}{\log i})(T_\mappend(m))$
for some $m$ that bounds the size of the monoids.

The time complexity of parallel algorithm consists on the base cost on running @f@
at each chunk and then parallel concatenating the $\frac{n}{j}$ chunks.
\begin{equation}
O((i-1)(\frac{\log n - \log j}{\log i})T_\mappend(m) + T_f(\frac{n}{j})) \label{eq:complexity}
\end{equation}
Since time complexity depends on the time complexity of @<>@
for the parallel algorithm to be efficient time complexity of @<>@ should be constant.
Otherwise, if it depends on the size of the input, the size of monoids can grow at each iteration of @mconcat@.

Moreover, from the complexity analysis we observe that time grows on bigger @i@ and smaller @j@.
Thus, chunking the input in small chunks while splitting the monoid list in half leads
to more parallelism, and thus (assuming infinite processors and no caching) greatest speedup.
}

\section{Case Study: Correctness of Parallel String Matching}\label{sec:stringmatching}

\S~\ref{sec:parallelization} showed that any monoid morphism
whose domain is chunkable can be parallelized. We now make use of that
result to parallelize string matching. We start by observing that
strings are a chunkable monoid. 
We then turn string matching for a
given target into a monoid morphism from a string to a suitable
monoid, @SM target@, defined in
\S~\ref{subsec:stringmatcher}. 
Finally, in \S~\ref{subsec:parallel-string-matching}, we parallelize string matching
by a simple use of the
parallel morphism function of \S~\ref{subsec:both-levels}. 

\ignore{
  In this section we apply the Correctness of Parallelization Theorem~\ref{theorem:two-level}
to a string matching function @toSM@
that is a morphism between strings and the indices where a target string appears,
to get correctness of parallelization of @toSM@.

We define @toSM :: RString -> SM target@
from a Refined String data type @RString@
to a dependently typed string maching data type @SM target@
where @target@ represents the substring to be matched.
To apply Theorem~\ref{theorem:two-level} on @toSM@ we need to discharge three proof obligations.
\begin{itemize}
\item @RString@ is a chunkable monoid (\S~\ref{subsec:refinedstrings}),
\item @SM target@ is a monoid (\S~\ref{subsec:stringmatcher}), and
\item @toSM@ is a morphism between @RString@ and @SM target@ (\S~\ref{subsec:smmorphism}).
\end{itemize}
With these proof obligations discharged we conclude (\S~\ref{subsec:parallel-string-matching})
correctness of parallel string matching.
}

\subsection{Refined Strings are Chunkable Monoids}\label{subsec:refinedstrings}
We define a new type @RString@, which is a chunkable monoid, to be the
domain of our string matching function. Our type simply wraps
Haskell's existing @ByteString@.
\begin{code}
data RString = RS BS.ByteString
\end{code}
Similarly, we wrap the existing @ByteString@ functions we will need to
show @RString@ is a chunkable monoid.
\begin{code}
stringMempty = RS (BS.empty)
(RS x) stringMappend (RS y)= S (x `BS.append` y)

lenStr    (RS x) = BS.length x
takeStr i (RS x) = RS (BS.take i x)
dropStr i (RS x) = RS (BS.take i x)
\end{code}
Although it is possible to explicitly prove that @ByteString@
implements a chunkable monoid~\cite{realworldliquid14}, it is time
consuming and orthogonal to our purpose. Instead, we
just \textit{assume} the chunkable monoid properties of @RString@--
thus demonstrating that refinement reflection is capable of doing
gradual verification.

\ignore{
We follow the easy route, defining the @RString@ data type to be a wrapper of the
optimized @ByteString@ and

This allows our implementation to \textit{use the optimized library
functions} of @ByteString@; additionally, it shows that refinement
reflection can be used for \textit{gradual verification} where
verified code uses untrusted components that are explicitely assumed
to satisfy required properties.
Proving that @ByteString@ implements a chunkable monoid is feasible,
as evidenced by the current verification of Bytestring functions~\cite{realworldliquid14},
but it is time consuming and orthogonal to the string matching proof.
%

We need to use the above operators to specify the chunkable monoid
laws, but we cannot reflect the operators in the logic, as the
ByteString functions do not exist in the logic.  We ``manually
reflect'' each of the above functions in the logic, by 1. defining a
logical uninterpreted function for each of them that assumed to be
equal to the Haskell functions and 2. use the logical functions to
specify the chunkable monoid laws.
}

For instance, we define a logical uninterpreted function
@stringMappend@ and relate it to the Haskell @stringMappend@ function
via an assumed (unchecked) type.
\begin{code}
assume (stringMappend)
  :: x:RString -> y:RString -> {v:RString | v = x stringMappend y}
\end{code}
Then, we use the uninterpreted function @stringMappend@ in the logic
to assume monoid laws, like associativity.
\begin{code}
assume assocStr :: x:RString -> y:RString -> z:RString
                 -> { x <+> (y <+> z) = (x <+> y) <+> z }
assocStr _ _     = trivial
\end{code}
Haskell applications of @stringMappend@ are interpreted in the logic
via the logical @stringMappend@ that satisfies associativity via theorem @assocStr@.

Similarly for the chunkable methods, we define the uninterpreted functions
@takeStr@, @dropStr@ and @lenStr@ in the logic,
and use them to strengthen the result types of the respective functions.
\ignore{
For example the type of @takeStr@ includes both the length specifications
from chunkable monoid and the uninterpreted function equality @v = takStr i x@.
\begin{code}
assume takeStr
  :: i:Nat -> x:{RString | i <= lenStr x}
  ->  {v:RString | lenStr v = i && v = takeStr i x }
\end{code}
We use the uninterpreted function @takeStr@ and @dropStr@ to
specify and \textit{assume} the @take@-@drop@ property of chunkable monoids.
\begin{code}
assume takeDropPropStr
 :: i:Nat -> x:RString -> {x = takeStr i x <+> dropStr i x}
takeDropPropStr _ _ = trivial
\end{code}
} With the above function definitions (in both Haskell and logic) and
assumed type specifications, Liquid Haskell will check (or rather
assume) that the specifications of chunkable monoid, as defined in the
Definitions~\ref{definition:monoid} and~\ref{definition:chunkable},
are satisfied.
We conclude with the assumption (rather that theorem)
that @RString@ is a chunkable monoid.
\begin{assumption}[RString is a Chunkable Monoid]\label{assumption:rstring}
(@RString@, @stringMempty@, @stringMappend@)
combined with the methods
@lenStr@, @takeStr@, @dropStr@ and @takeDropPropStr@
is a chunkable monoid.
\end{assumption}

\subsection{String Matching Monoid}\label{subsec:stringmatcher}
String matching is determining all the indices in a source string
where a given target string begins; for example, for source string
\texttt{ababab} and target \texttt{aba} the results of string
matching would be \texttt{[0, 2]}. 
\ignore{
We will use the following code as a specification of string matching:
\begin{code}
stringMatcher inp trg = go 0 where
  trgLen = length trg
  go i | i >= length inp = []
       | take trgLen (drop i inp) == trg = i : rest
       | otherwise = rest
  where rest = go (i + 1)
\end{code}
Though this implementation is not efficient, it is simple and provides
a good starting point from which to develop the machinery in this section.
}

We now define a suitable monoid, @SM target@, for the codomain of
a string matching function, where @target@ is the string being looked
for.
Additionally, we will define a function @toSM :: RString -> SM target@
which does the string matching and is indeed a monoid morphism from
@RString@ to @SM target@ for a given @target@.

\ignore{
Next, we use @RString@ to define the string matching data type
@SM target@ that stores an input refined string and all the
indices where the string type literal @target@ appears in the input.
We defined the monoid methods @mempty@ and @mappend@ of the string matcher
and prove that the structure (@SM target@, @mempty@, @mappend@) is a monoid.
}
\subsubsection{String Matching Monoid}

We define the data type
@SM target@ to contain a refined string field @input@ and
a list of all the indices in @input@ where the
@target@ appears.
\begin{code}
data SM (target :: Symbol) where
  SM :: input:RString
     -> indices:[GoodIndex input target]
     -> SM target
\end{code}
We use the string type literal~\footnote{\texttt{Symbol} is a kind and
target is effectively a singleton type.} to parameterize the monoid
over the target being matched. This encoding allows the type checker
to statically ensure that only searches for the same target can be
merged together.  The input field is a refined string, and the indices
field is a list of good indices.  For simplicity we present lists as
Haskell's built-in lists, but our implementation uses the reflected
list type, @L@, defined in \S~\ref{sec:haskell-proofs}.

A @GoodIndex input target@ is a refined type alias for a natural
number @i@ for which @target@ appears at position @i@ of @input@.  As
an example, the good indices of @"abcab"@ on @"ababcabcab"@ are
@[2,5]@.
\begin{code}
type GoodIndex Input Target
  = {i:Nat | isGoodIndex Input (fromString Target) i }

isGoodIndex :: RString -> RString -> Int -> Bool
isGoodIndex input target i
  = (subString i (lenStr target) input  == target)
  && (i + lenStr target <= lenStr input)

subString :: Int -> Int -> RString -> RString
subString o l = takeStr l . dropStr o
\end{code}
\ignore{
\begin{code}
goodSM :: SM "abcab"
goodSM = SM "ababcabcab" [2, 5]

badSM  :: SM "abcab"
badSM  = SM "ababcabcab" [0, 7]
\end{code}
\ignore{
\NV{Liquid Haskell actually will reject both the above, as the lenStr and subString functions are uninterpreted}
}
}

\subsubsection{Monoid Methods for String Matching}~\label{subsec:monoid:methods}
Next, we define the mappend and identity elements for string matching.

The \textit{identity element} @mempty@ of @SM t@, for each target @t@, is
defined to contain the identity @RString@ (@stringMempty@) and the
identity @List@ (@listMempty@).
\begin{code}
mempty:: forall (t :: Symbol). SM t
mempty = SM stringMempty listMempty
\end{code}

\ignore{
The associative operator, @(mappend)@,
appends the two input strings.
The appended indices, as depicted in Figure~\ref{fig:mappend:indices},
are the concatenations of three list indices:
\begin{enumerate}
\item The indices @xis@ of the first input, casted to good indices in the new structure,
\item the new indices @xyis@ created when concatenating the two strings, and
\item the indices @yis@ of the second input, shifted right @lenStr y@ units.
\end{enumerate}
}
\begin{figure}[t]
\includegraphics[scale=0.5]{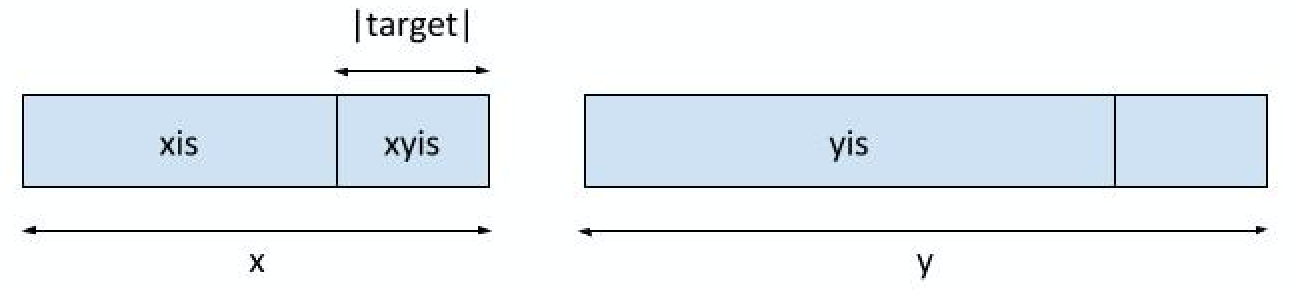}
\caption{Mappend indices of String Matcher}
\label{fig:mappend:indices}
\end{figure}
The Haskell definition of @<>@, the monoid operation for @SM t@, is as follows.
\begin{code}
(mappend)::forall (t::Symbol). KnownSymbol t => SM t -> SM t -> SM t
(SM x xis) mappend (SM y yis)
  = SM (x stringMappend y) (xis' listMappend xyis listMappend yis')
  where
    tg   = fromString (symbolVal (Proxy :: Proxy t))
    xis' = map (castGoodIndexLeft tg x y) xis
    xyis = makeNewIndices x y tg
    yis' = map (shiftStringRight tg x y) yis
\end{code}
Note again that capturing target as a type parameter is critical,
otherwise there is no way for the Haskell's type system to specify
that both arguments of @(mappend)@ are string matchers on the same target.

The action of @(<>)@ on the two @input@ fields is straightforward;
however, the action on the two @indices@ is complicated by the need to
shift indices and the possibility of new matches arising from the
concatenation of the two @input@
fields. Figure~\ref{fig:mappend:indices} illustrates the three pieces
of the new @indices@ field which we now explain in more detail.

\paragraph{1. Casting Good Indices}
If @xis@ is a list of good indices for the string @x@ and the target
@tg@, then @xis@ is also a list of good indices for the string
@x stringMappend y@ and the target @tg@, for each @y@.
To prove this property we need to invoke the property
@subStrAppendRight@ on Refined Strings that establishes
substring preservation on string right appending.
\begin{code}
assume subStrAppendRight
    :: sl:RString -> sr:RString -> j:Int
    ->  i:{Int | i + j <= lenStr sl }
    ->  { subString sl i j = subString (sl stringMappend sr) i j }
\end{code}
The specification of @subStrAppendRight@ ensures that for each
string @sl@ and @sr@ and each integer @i@ and @j@ whose sum is within @sl@,
the substring from @i@ with length @j@ is identical in @sl@ and in @(sl stringMappend sr)@.
The function @castGoodIndexLeft@ applies the above property to an index @i@
to cast it from a good index on @sl@ to a good index on @(sl stringMappend sr)@
\begin{code}
castGoodIndexLeft
  :: tg:RString -> sl:RString -> sr:RString
  -> i:GoodIndex sl tg
  -> {v:GoodIndex (sl stringMappend sr) target | v = i}

castGoodIndexLeft tg sl sr i
  = cast (subStrAppendRight sl sr (lenStr tg) i) i
\end{code}
Where @cast p x@ returns @x@, after enforcing the properties of @p@ in the logic
\begin{code}
cast :: b -> x:a -> {v:a | v = x }
cast _ x = x
\end{code}
Moreover, in the logic, each expression @cast p x@
is reflected as @x@,
thus allowing random (\ie non-reflected) Haskell expressions to appear in @p@.

\paragraph{2. Creation of new indices}
The concatenation of two input strings @sl@ and @sr@
may create new good indices.
For instance, concatenation of
@"ababcab"@ with @"cab"@
leads to a new occurence of @"abcab"@ at index @5@ which
does not occur in either of the two input strings.
These new good indices can appear only at the last @lenStr tg@ positions
of the left input @sl@.
@makeNewIndices sl sr tg@ detects all such good new indices.
\begin{code}
makeNewIndices
  :: sl:RString -> sr:RString -> tg:RString
  -> [GoodIndex {sl stringMappend sr} tg]
makeNewIndices sl sr tg
  | lenStr tg < 2 = []
  | otherwise     = makeIndices (sl stringMappend sr) tg lo hi
  where
    lo = maxInt (lenStr sl - (lenStr tg - 1)) 0
    hi = lenStr sl - 1
\end{code}
If the length of the @tg@ is less than 2, then no new good indices are created.
Otherwise,
the call on @makeIndices@ returns all the good indices of the input
@sl stringMappend sr@ for target @tg@
in the range from @maxInt (lenStr sl-(lenStr tg-1)) 0@ to  @lenStr sl-1@.

Generally, @makeIndices s tg lo hi@ returns the good indices
of the input string @s@ for target @tg@ in the range from @lo@ to @hi@.
\begin{code}
makeIndices
  :: s:RString -> tg:RString -> lo:Nat
  -> hi:Int -> [GoodIndex s tg]
makeIndices s tg lo hi
  | hi < lo             = []
  | isGoodIndex s tg lo = lo:rest
  | otherwise           = rest
  where
    rest = makeIndices s tg (lo + 1) hi
\end{code}

It is important to note that
@makeNewIndices@ does not scan all the input,
instead only searching at most @lenStr tg@ positions for new good indices.
Thus, the time complexity to create the new indices is linear
on the size of the target but independent of the size of the input.

\paragraph{3. Shift Good Indices}
If @yis@ is a list of good indices on the string @y@ with target @tg@,
then we need to shift each element of @yis@ right @lenStr x@ units to
get a list of good indices for the string @x stringMappend y@.

To prove this property we need to invoke the property
@subStrAppendLeft@ on Refined Strings that establishes
substring shifting on string left appending.
\begin{code}
assume subStrAppendLeft
  :: sl:RString -> sr:RString
  -> j:Int -> i:Int
  -> {subStr sr i j = subStr (sl stringMappend sr) (lenStr sl+i) j}
\end{code}
The specification of @subStrAppendLeft@ ensures that for each
string @sl@ and @sr@ and each integers @i@ and @j@,
the substring from @i@ with length @j@ on @sr@
is equal to the substring from @lenStr sl + i@
with length @j@ on @(sl stringMappend sr)@.
The function @shiftStringRight@ both shifts the input index @i@ by @lenStr sl@
and applies the @subStrAppendLeft@ property to it,
casting @i@ from a good index on @sr@ to a good index on @(sl stringMappend sr)@

Thus, @shiftStringRight@ both appropriately shifts the index
and casts the shifted index using the above theorem:
\begin{code}
shiftStringRight
  :: tg:RString -> sl:RString -> sr:RString
  -> i:GoodIndex sr tg
  -> {v:(GoodIndex (sl stringMappend sr) tg) | v = i + lenStr sl}
shiftStringRight tg sl sr i
  = subStrAppendLeft sl sr (lenStr tg) i
    `cast` i + lenStr sl
\end{code}

\subsubsection{String Matching is a Monoid}
Next we prove that the monoid methods @mempty@ and @(mappend)@ satisfy
the monoid laws.
\begin{theorem}[SM is a Monoid]\label{theorem:stringmatchers}
(@SM t@, @mempty@, @mappend@)
is a monoid.
\end{theorem}
\begin{proof}
According to the Monoid Definition~\ref{definition:monoid},
we prove that string matching is a monoid,
by providing safe implementations for the monoid law functions.
First, we prove \textit{left identity}.
\begin{code}
idLeft :: x:SM t -> {mempty mappend x = xs }
idLeft (SM i is)
  =   (mempty :: SM t) mappend (SM i is)
  ==. (SM stringMempty listMempty) mappend (SM i is)
  ==. SM (stringMempty <+> i) (is1 ++ isNew ++ is2)
       ? idLeftStr i
  ==. SM i ([] ++ [] ++ is)
       ? (mapShiftZero tg i is && newIsNullRight i tg)
  ==. SM i is
       ? idLeftList is
  *** QED
  where
    tg    = fromString (symbolVal (Proxy :: Proxy t))
    is1   = map (castGoodIndexRight tg i stringMempty) []
    isNew = makeNewIndices stringMempty i tg
    is2   = (map (shiftStringRight tg stringMempty i) is)
\end{code}

The proof proceeds by rewriting, using left identity of the monoid strings and lists,
and two more lemmata.
\begin{itemize}
\item Identity of shifting by an empty string.
\begin{code}
mapShiftZero :: tg:RString -> i:RString
  -> is:[GoodIndex i target]
  -> {map (shiftStringRight tg stringMempty i) is = is }
\end{code}
The lemma is proven by induction on @is@ and
the assumption that empty strings have length 0.
\item No new indices are created.
\begin{code}
newIsNullLeft :: s:RString -> t:RString
  -> {makeNewIndices stringMempty s t = [] }
\end{code}
The proof relies on the fact that @makeIndices@
is called on the empty range from @0@ to @-1@
and returns @[]@.
\end{itemize}

Next, we prove \textit{right identity}.
\begin{code}
idRight :: x:SM t -> {x mappend mempty = x }
idRight (SM i is)
  =   (SM i is) mappend (mempty :: SM t)
  ==. (SM i is) mappend (SM stringMempty listMempty)
  ==. SM (i stringMappend stringMempty) (is1 listMappend isNew listMappend is2)
      ? idRightStr i
  ==. SM i (is listMappend N listMappend N)
      ? (mapCastId tg i stringMempty is && newIsNullLeft i tg)
  ==. SM i is
      ? idRightList is
  ***  QED
  where
    tg    = fromString (symbolVal (Proxy :: Proxy t))
    is1   = map (castGoodIndexRight tg i stringMempty) is
    isNew = makeNewIndices i stringEmp tg
    is2   = map (shiftStringRight tg i stringMempty) []
\end{code}
The proof proceeds by rewriting,
using right identity on strings and lists and two more lemmata.
\begin{itemize}
\item Identity of casting is proven
\begin{code}
mapCastId :: tg:RString -> x:RString -> y:RString
  -> is:[GoodIndex x tg] ->
  -> {map (castGoodIndexRight tg x y) is = is}
\end{code}
We prove identity of casts by induction on @is@ and
identity of casting on a single index.
\item No new indices are created.
\begin{code}
newIsNullLeft :: s:RString -> t:RString
  -> {makeNewIndices s stringMempty t = listMempty }
\end{code}
The proof proceeds by case splitting
on the relative length of @s@ and @t@.
At each case we prove by induction that all
the potential new indices would be out of bounds and thus
no new good indices would be created.
\end{itemize}

- Finally we prove \textit{associativity}.
For space, we only provide a proof sketch.
The whole proof is available online~\cite{implementation}.
Our goal is to show equality of the left and right associative string matchers.
\begin{code}
assoc :: x:SM t -> y:SM t -> z:SM t
       -> { x mappend (y mappend z) = (x mappend y) mappend z}
\end{code}
To prove equality of the two string matchers we show
that the input and indices fields are respectively equal.
Equality of the input fields follows by associativity of RStrings.
Equality of the index list proceeds in three steps.
\begin{enumerate}
\item Using list associativity and distribution of index shifting,
we group the indices in the five lists shown in
Figure~\ref{fig:mappend:assoc}: the indices of the input @x@, the new
indices from mappending @x@ to @y@, the indices of the input @y@, the
new indices from mappending @x@ to @y@, and the indices of the input
@z@.
\item The representation of each group depends on the order of appending.
For example, if @zis1@ (resp. @zis2@) is the group @zis@ when
right (resp. left) mappend happened first, then we have
\begin{code}
zis1 = map (shiftStringRight tg xi (yi stringMappend zi))
           (map (shiftStringRight tg yi zi) zis)

zis2 = map (shiftStringRight tg (xi stringMappend yi) zi) zis
\end{code}
That is, in right first, the indices of @z@ are first shifted
by the length of @yi@ and then by the length of @xi@,
while in the left first case, the indices of @z@ are shifted by the
length of @xi stringMappend yi@.
In this second step of the proof we prove, using lemmata,
the equivalence of the different group representations.
The most interesting lemma we use is called @assocNewIndices@ and proves
equivalence of all the three middle groups together
by case analysis on the relative lengths of the target @tg@ and the middle string @yi@.
\item After proving equivalence of representations,
we again use list associativity and distribution of casts to wrap the
index groups back in string matchers.
\end{enumerate}
\begin{figure}
\includegraphics[scale=0.5]{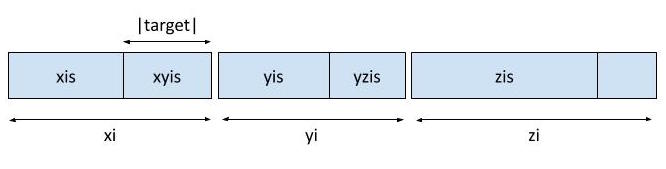}
\caption{Associativity of String Matching}
\label{fig:mappend:assoc}
\end{figure}
We now sketch the three proof steps, while the whole proof
is available online~\cite{implementation}.
\begin{code}
assoc x@(SM xi xis) y@(SM yi yis) z@(SM zi zis)
  -- Step 1: unwrapping the indices
  =   x <> (y <> z)
  ==. (SM xi xis) <> ((SM yi yis) <> (SM zi zis))
                         ...
  -- via list associativity and distribution of shifts
  ==. SM i (xis1 ++ ((xyis1 ++ yis1 ++ yzis1) ++ zis1))
  -- Step 2: Equivalence of representations
  ==. SM i (xis2 ++ ((xyis1 ++ yis1 ++ yzis1) ++ zis1))
      ? castConcat tg xi yi zi xis
  ==. SM i (xis2 ++ ((xyis1 ++ yis1 ++ yzis1) ++ zis2))
      ? mapLenFusion tg xi yi zi zis
  ==. SM i (xis2 ++ ((xyis2 ++ yis2 ++ yzis2) ++ zis2))
      ? assocNewIndices y tg xi yi zi yis
  -- Step 3: Wrapping the indices
                         ...
  -- via list associativity and distribution of casts
  ==. (SM xi xis <> SM yi yis) <> SM zi zis
  =   (x <> y) <> z
  *** QED
  where
    i     = xi stringMappend (yi stringMappend zi)

    yzis1 = map (shiftStringRight tg xi (yi <+> zi)) yzis
    yzis2 = makeNewIndices (xi <+> yi) zi tg
    yzis  = makeNewIndices yi zi tg
    ...
\end{code}
\qed\end{proof}

\subsection{String Matching Monoid Morphism}\label{subsec:smmorphism}
Next, we define the function @toSM :: RString -> SM target@ which does
the actual string matching computation for a set
target~\footnote{\texttt{toSM} assumes the target is clear from the
  calling context; it is also possible to write a wrapper function
  taking an explicit target which gets existentially reflected into
  the type.}
\ignore{ The function @toSM input@ creates a string matcher @SM
  target@ structure that matches the type level symbol target to the
  string value argument @input@.  }
\begin{code}
toSM :: forall (target :: Symbol). (KnownSymbol target)
     => RString -> SM target
toSM input = SM input (makeSMIndices input tg) where
  tg = fromString (symbolVal (Proxy :: Proxy target))

makeSMIndices
  :: x:RString -> tg:RString -> [GoodIndex x tg]
makeSMIndices x tg
  = makeIndices x tg 0 (lenStr tg - 1)
\end{code}
The input field of the result is the input string;
the indices field is computed by calling @makeIndices@
within the range of the @input@, that is from
@0@ to @lenStr input - 1@.
\begin{code}
\end{code}

We now prove that @toSM@ is a monoid morphism.
\begin{theorem}[$\texttt{toSM}$ is a Morphism]\label{theorem:smmorphism}
@toSM :: RString -> SM t@ is a morphism between the monoids
(@RString@, @stringMempty@, @stringMappend@) and (@SM t@, @mempty@, @mappend@).
\end{theorem}
\begin{proof}
Based on definition~\ref{definition:morphism}, proving @toSM@
is a morphism requires constructing a valid inhabitant of the type
\begin{code}
Morphism RString (SM t) toSM
  = x:RString -> y:RString
  -> {toSM stringMempty = mempty && toSM (x <+> y) = toSM x <> toSM y}
\end{code}
We define the function @distributestoSM :: Morphism RString (SM t) toSM@
to be the required valid inhabitant.

The core of the proof starts from
exploring the string matcher @toSM x <> toSM y@.
This string matcher contains three sets of indices
as illustrated in Figure~\ref{fig:mappend:indices}:
(1) @xis@ from the input @x@,
(2) @xyis@ from appending the two strings, and
(3) @yis@ from the input @y@.
We prove that appending these three groups of indices together gives
exactly the good indices of @x stringMappend y@, which are also the
value of the indices field in the result of
@toSM (x stringMappend y)@.
\begin{code}
distributestoSM x y
  =   (toSM x :: SM target) <> (toSM y :: SM target)
  ==. (SM x is1) <> (SM y is2)
  ==. SM i (xis ++ xyis ++ yis)
  ==. SM i (makeIndices i tg 0 hi1 ++ yis)
      ? (mapCastId tg x y is1 && mergeNewIndices tg x y)
  ==. SM i (makeIndices i tg 0       hi1
         ++ makeIndices i tg (hi1+1) hi)
      ? shiftIndicesRight 0 hi2 x y tg
  ==. SM i is
      ? mergeIndices i tg 0 hi1 hi
  ==. toSM (x <+> y)
  *** QED
  where
    xis  = map (castGoodIndexRight tg x y) is1
    xyis = makeNewIndices x y tg
    yis  = map (shiftStringRight   tg x y) is2
    tg   = fromString (symbolVal (Proxy::Proxy target))
    is1  = makeSMIndices x tg
    is2  = makeSMIndices y tg
    is   = makeSMIndices i tg
    i    = x <+> y
    hi1  = lenStr x - 1
    hi2  = lenStr y - 1
    hi   = lenStr i - 1
\end{code}
The most interesting lemma we use is
@mergeIndices x tg lo mid hi@
that states that for the input @x@ and the target @tg@
if we append the indices in the range  from @to@ to @mid@
with the indices in the range from @mid+1@ to @hi@,
we get exactly the indices in the range from @lo@ to @hi@.
This property is formalized in the type of the lemma.
\begin{code}
mergeIndices
 :: x:RString -> tg:RString
  -> lo:Nat -> mid:{Int | lo <= mid} -> hi:{Int | mid <= hi}
  -> {   makeIndices x tg lo hi
     =  makeIndices x tg lo mid
     ++ makeIndices x tg (mid+1) hi}
\end{code}
The proof proceeds by induction on @mid@ and using three more lemmata:
\begin{itemize}
\item @mergeNewIndices@ states that appending the indices @xis@ and @xyis@
is equivalent to the good indices of @x stringMappend y@ from @0@ to @lenStr x - 1@.
The proof case splits on the relative sizes of @tg@ and @x@
and is using @mergeIndices@ on @mid = lenStr x1 - lenStr tg@
in the case where @tg@ is smaller than @x@.
\item @mapCastId@ states that casting a list of indices returns the same list.
\item @shiftIndicesRight@ states that shifting right @i@ units the indices from @lo@ to @hi@
is equivalent to computing the indices from @i + lo@
to @i + hi@ on the string @x stringMappend y@, with @lenStr x = i@.
\end{itemize}
\qed\end{proof}

\subsection{Parallel String Matching}\label{subsec:parallel-string-matching}
We conclude this section with
the definition of a parallelized version of string matching.
We put all the theorems together to prove 
that the sequential and parallel versions always give the same result.

We define @toSMPar@ as a parallel version of @toSM@ using machinery of section~\ref{sec:parallelization}.
\begin{code}
toSMPar :: forall (target :: Symbol). (KnownSymbol target)
        => Int -> Int -> RString -> SM target
toSMPar i j = pmconcat i . pmap toSM . chunkStr j
\end{code}
First, @chunkStr@ splits the input into @j@ chunks.
Then, @pmap@ applies @toSM@ at each chunk in parallel.
Finally, @pmconat@ concatenates the mapped chunks in parallel
using @mappend@, the monoidal operation for @SM target@.

\paragraph{Correctness.}
We prove correctness of @toSMPar@ directly from
Theorem~\ref{theorem:two-level}.
\begin{theorem}[Correctness of Parallel String Matching]\label{theorem:correctness}
For each parameter @i@ and @j@, and input @x@,
@toSMPar i j x@ is always equal to @toSM x@.
\begin{code}
correctness :: i:Int -> j:Int -> x:RString
             -> {toSM x = toSMPar i j x}
\end{code}
\end{theorem}

\begin{proof}
The proof follows by direct application of Theorem~\ref{theorem:two-level}
on the chunkable monoid (@RString@, @$\eta$@, @<+>@) (by Assumption~\ref{assumption:rstring})
and the monoid (@SM t@, @$\epsilon$@, @<>@) (by Theorem~\ref{theorem:stringmatchers}).
\begin{code}
correctness i j x
  =   toSMPar i j x
  ==. pmconcat i (pmap toSM (chunkStr j x))
  ==. toSM is
    ? parallelismEquivalence toSM distributestoSM x i j
  *** QED
\end{code}
Note that application of the theorem @parallelismEquivalence@
requires a proof that its first argument @toSM@ is a morphism.
By Theorem~\ref{theorem:monoid:distribution},
the required proof is provided as the function @distributestoSM@.
\qed\end{proof}

\ignore{

\paragraph{Time Complexity.}
Counting only string comparisons as the expensive operations,
the sequential string matcher on input @x@ runs in time
linear to @n = lenStr x@. Thus $T_\texttt{toSM}(n) = O(n)$.

We get time complexity of @toSMPar@ by the time complexity of
two-level parallel algorithms equation~\ref{eq:complexity},
with the time of string matching mappend being linear on the length
of the target @t = lenStr tg@, or
$T_\mappend(\texttt{SM})= O(t)$.
$$
T_\texttt{toSMPar} (n, t, i, j) =
O((i-1)(\frac{\log n - \log j}{\log i}) t  + \frac{n}{j})
$$
The above analysis refers to a model with infinite processor and no caching.
To compare the algorithms in practice,
we matched the target "the"
in  Oscar Wilde's "The Picture of Dorian Gray", a text of @n = 431372@ characters
using a two processor Intel Core i5.
The sequential algorithm detected 4590 indices in 40 ms.
We experimented with different parallization factors @i@ and chunk sizes @j / n@
and observed up to $50\%$ speedups of the parallel algorithm for parallelization factor
@4@ and @8@ chunks.
As a different experiment, we matched the input against its size @t = 400@ prefix,
a size comparable to the input size @n@.
For bigger targets,
mappend gets slower, as it has complexity linear to the size of target.
We observed $20\%$ speedups for @t=400@ target but also $30\%$ slow downs for various sizes of @i@ and @j@.
In all cases the indices returned by the sequential and the parallel algorithms were the same.
}

\section{Evaluation: Strengths \& Limitations}\label{sec:evaluation}

Verification of Parallel String Matching is the first realistic
proof that uses (Liquid) Haskell
to prove properties \textit{about} program functions.
In this section we use the String Matching proof
to quantitatively and qualitatively evaluate theorem proving in Haskell.

\paragraph{Quantitative Evaluation.}
The Correctness of Parallel String Matching proof
can be found online~\cite{implementation}.
Verification time, that is the time Liquid Haskell needs to check the proof,
is 75 sec on a dual-core Intel Core i5-4278U processor.
The proof consists of \textit{1839} lines of code.
Out of those
\begin{itemize}
\item \textit{226} are Haskell ``runtime'' code,
\item \textit{112} are liquid comments on the ``runtime'' Haskell code,
\item \textit{1307} are Haskell proof terms, that is functions with @Proof@ result type, and
\item \textit{194}  are liquid comments to specify theorems.
\end{itemize}
Counting both liquid comments and Haskell proof terms as verification code,
we conclude that the proof requires 7x the lines of ``runtime'' code.
This ratio is high and takes us to 2006 Coq,
when Leroy~\cite{Leroy06formalcertification} verified
the initial CompCert C compiler with
the ratio of verification to compiler lines being 6x.

\paragraph{Strengths.}
Though currently verbose,
deep verification using Liquid Haskell has many benefits.
First and foremost,
\textit{the target code is written in the general purpose Haskell}
and thus can use advanced Haskell features, including
type literals, deriving instances, inline annotations
and optimized library functions like @ByteString@.
Even diverging functions can coexist with the target code, as long
as they are not reflected into logic~\cite{Vazou14}.

Moreover, \textit{SMTs are used to automate the proofs}
over key theories like linear arithmetic and equality.
As an example, associativity of @(+)@
is assumed throughout the proofs while shifting indices.
Our proof could be further automated 
by mapping refined strings to SMT strings and  
using the automated SMT string theory.
We did not follow this approach because we want to show that
our techinique can be used to prove any (and not only domain specific)
program properties.

Finally, we get further automation via
\textit{Liquid Type Inference}~\cite{LiquidPLDI08}.
Properties about program functions,
expressed as type specifications with unit result,
often depend on program invariants,
expressed as vanilla refinement types, and vice versa.
For example, we need the invariant that all indices of
a string matcher are good indices
to prove associativity of @(mappend)@.
%
Even though Liquid Haskell cannot currently synthesize proof terms,
it performs really well at inferring and propagating program invariants (like good indices)
via the abstract interpretation framework of Liquid Types.

\paragraph{Limitations.}
There are severe limitations that should be addressed
to make theorem proving in Haskell a pleasant and usable technique.
As mentioned earlier \textit{the proofs are verbose}.
There are a few cases where the proofs require domain specific knowledge.
For example, to prove associativity of string matching
@x mappend (y mappend z) = (x mappend y) mappend z@
we need a theorem that performs case analysis on the relative length of
the input field of @y@ and the target string.
Unlike this case split though, most proofs
do not require domain specific knowledge and merely proceed
by term rewriting and structural inductuction
that should be automated
via Coq-like~\cite{coq-book} tactics or/and Dafny-like~\cite{dafny} heuristics.
For example, synquid~\cite{polikarpova16} could be used to automatically
synthesize proof terms.

Currently, we suffer from two engineering limitations.
First, all reflected function should exist in the same module,
as reflection needs access to the function implementation
that is unknown for imported functions.
This is the reason why we need to use a user defined,
instead of Haskell's built-in, list.
In our implementation we used @CPP@ as a current workaround
of the one module restriction.
Second, class methods
cannot be currently reflected.
Our current workaround is to define Haskell functions instead
of class instances.
For example (@append@, @nil@) and (@concatStr@, @emptyStr@)
define the monoid methods of List and Refined String respectively.

Overall, we believe that the strengths outweigh the limitations which
will be addressed in the near future,
rendering Haskell a powerful theorem prover.

\section{Related Work}\label{sec:related}

\paragraph{SMT-Based Verification}
SMT solvers have been extensively used to automate
reasoning on verification languages like
Dafny~\cite{dafny}, Fstar~\cite{fstar} and Why3~\cite{why3}.
These languages are designed for verification,
thus have limited support for commonly used language
features like parallelism and optimized libraries
that we use in our verified implementation.
Refinement Types~\cite{ConstableS87,FreemanPfenningDONTCITE91,Rushby98}
on the other hand, target existing general purpose languages,
such as
ML~\cite{pfenningxi98,GordonRefinement09,LiquidPLDI08},
C~\cite{deputy,LiquidPOPL10},
Haskell~\cite{Vazou14},
Racket~\cite{RefinedRacket}
and Scala~\cite{refinedscala}.
However, before Refinement Reflection~\cite{reflection} was introduced,
they only allowed ``shallow'' program specifications,
that is, properties that only talk about behaviors of program functions
but not functions themselves.

\paragraph{Dependent Types}
Unlike Refinement Types, dependent type systems,
like Coq~\cite{coq-book}, Adga~\cite{agda} and Isabelle/HOL~\cite{isabelle} allow for ``deep'' specifications
which talk about program functions,
such as the program equivalence reasoning we presented.
Compared to (Liquid) Haskell,
these systems allow for tactics and heuristics
that automate proof term generation
but lack SMT automations and
general-purpose language features,
like non-termination, exceptions and IO.
Zombie~\cite{Zombie,Sjoberg2015} and Fstar~\cite{fstar} allow dependent types to
coexist with divergent and effectful programs,
but still lack the optimized libraries,
like @ByteSting@, which come
with a general purpose language
with long history, like Haskell.

\paragraph{Parallel Code Verification}
Dependent type theorem provers have been used before to
verify parallel code.
BSP-Why~\cite{bspwhy} is an extension to Why2 that is using both Coq and SMTs
to discharge user specified verification conditions.
Daum~\cite{daum07} used Isabelle to formalize the semantics
of a type-safe subset of C, 
by extending Schirmer's~\cite{schirmer06}
formalization of sequential imperative languages.
Finally, Swierstra~\cite{wouter10} formalized mutable arrays in Agda
and used them to reason about distributed maps and sums.

One work  closely related to ours is
SyDPaCC~\cite{SyDPaCC}, a Coq library that
automatically parallelizes list homomorphisms
by extracting parallel Ocaml versions of user provided Coq functions.
Unlike SyDPaCC, we are not automatically generating the parallel
function version, because of engineering limitations
(\S~\ref{sec:evaluation}).  Once these are addressed, code extraction
can be naturally implemented by turning
Theorem~\ref{theorem:two-level} into a Haskell type class with a
default parallelization method.
SyDPaCC used maximum prefix sum as a case study,
whose morphism verification is
much simpler than our string matching case study.
Finally, our implementation consists of
2K lines of Liquid Haskell, which we consider verbose and aim to
use tactics to simplify.
On the contrary, the SyDPaCC implementation
requires three different languages:
2K lines of Coq with tactics, 600 lines of Ocaml and 120 lines of C,
and is considered ``very concise''.

\section{Conclusion}\label{sec:conclusion}
We made the first non-trivial use of (Liquid) Haskell as a proof
assistant. 
We proved the parallelization of chunkable monoid
morphisms to be correct
and applied our parallelization technique to string matching,
resulting in a formally verified parallel string matcher.
Our proof uses refinement types to specify
equivalence theorems,
Haskell terms to express proofs,
and Liquid Haskell to check that the terms prove the theorems.
Based on our 1839LoC sophisticated proof we conclude that
Haskell can be successfully used as a theorem prover
to prove arbitrary theorems about real Haskell code
using SMT solvers to automate proofs
over key theories like linear arithmetic and equality.
However, Coq-like tactics or Dafny-like heurestics are required
to ease the user from manual proof term generation.

{
\bibliographystyle{plain}
\bibliography{sw}
}

\end{document}